\documentclass{eptcs}

\usepackage[utf8]{inputenc}
\usepackage{float}

\usepackage{amsthm}
\usepackage[all]{xy}
\usepackage{amsmath}
\usepackage{amssymb}
\usepackage{amsfonts}

\usepackage{enumerate}

\newtheorem{theorem}{Theorem}
\newtheorem{proposition}{Proposition}
\newtheorem{lemma}{Lemma}
\theoremstyle{definition}
\newtheorem{definition}{Definition}

\theoremstyle{definition}

\theoremstyle{definition}

\theoremstyle{definition}
\newtheorem{claim}{Claim}
\theoremstyle{definition}

\newcommand{\psf}{\mathcal{P}}

\newcommand{\game}{\mathcal{G}}
\newcommand{\plrone}{\mathtt{A}}
\newcommand{\plrtwo}{\mathtt{B}}
\newcommand{\eloi}{\mathtt{A}}
\newcommand{\abel}{\mathtt{B}}
\newcommand{\player}{\mathtt{P}}
\newcommand{\adversary}{\overline{\mathtt{P}}}
 \newcommand{\tset}[1]{{[\![ #1 ]\!]}}

\newcommand{\mdl}{\mathfrak{M}}

\newcommand{\Gal}{[\mathtt{A}]}
\newcommand{\Gbo}{[\mathtt{B}]}
\newcommand{\Gpl}{[\mathtt{P}]}
\newcommand{\gal}{[\mathtt{A}]}
\newcommand{\gbo}{[\mathtt{B}]}
\newcommand{\gpl}{[\player]}

\newcommand{\model}{\mathfrak{M}}
\newcommand{\statekp}{(\widehat{\varphi},k+1)}

\newcommand{\bisimilar}{\;\underline{\longleftrightarrow}\;}

\title{A New Game Equivalence and its Modal Logic}
\author{Johan van Benthem \and Nick Bezhanishvili \and Sebastian Enqvist}

\title{A New Game Equivalence and its Modal Logic}
\author{Johan van Benthem
\institute{Institute for Logic, Language and Computation \\
University of Amsterdam, Netherlands \\
Department of
Philosophy, Stanford University, USA   \\
 Changjiang Scholar Program, Tsinghua University, China
} 
\email{J.vanBenthem@uva.nl}
\and
Nick Bezhanishvili 
\institute{Institute for Logic, Language and Computation \\
University of Amsterdam, Netherlands} 
\email{N.Bezhanishvili@uva.nl}
\and
Sebastian Enqvist
\institute{Department for Philosophy \\
Stockholm University, Sweden}
\email{thesebastianenqvist@gmail.com}
}
\def\titlerunning{A New Game Equivalence}
\def\authorrunning{J. van Benthem, N. Bezhanishvili \& S. Enqvist}

\begin{document}
\maketitle

\begin{abstract}
We revisit the crucial issue of natural game equivalences, and semantics of game logics based on these. We present reasons for investigating finer concepts of game equivalence than equality of standard powers, though staying short of modal bisimulation. Concretely, we propose a more fine-grained notion of equality of  `basic powers' which record what players can force plus what they leave to others to do, a crucial feature of interaction. This notion is closer to game-theoretic strategic form, as we explain in detail, while remaining amenable to logical analysis. We determine the properties of basic powers via a new representation theorem, find a matching `instantial neighborhood game logic', and show how our analysis can be extended to a new game algebra and dynamic game logic.
\end{abstract}

\section{Introduction } 

Games are a basic model for interactive agency, but how much structure do we want to consider? Game theory offers strategic form games and extensive games, which represent two levels of structure, less or more detailed.
Logic of games has also looked at other natural invariances between representations of games, such as equivalence of powers for players. As in other areas of mathematics, the search for natural invariances continues, and in this paper we offer a new notion bridging between game theory and logic: strong power equivalence, that uses powers encoding a sort of qualitative equilibria. We determine its properties in a new representation theorem for the ``basic powers'' in a game, show that it has a natural associated logic, and that it supports an interesting new game algebra where the methodological principle of compositionality eventually forces us to change from functional to relational strategies. Besides the representation theorem for basic powers, the main technical contribution of the paper is a completeness theorem for the new game logic that we define. The proof uses a technique developed in \cite{vB16}, but requires a non-trivial adaptation due to the presence of extra frame constraints.

We believe that our proposed 
game equivalence is new, but even so, it fits with a body of 
earlier work. Our approach is partly inspired by the extensive
computational literature on process equivalences, ranging 
from coarser trace equivalence to more fine-grained notions
of bisimulation \cite{bergstra2001handbook}. Even more central in our approach was the
by now standard notion of power equivalence, implicit in the 
game algebra of Parikh \cite{parikh85}, which also links with the set-
theoretic forms for games presented in \cite{bonanno1992set}. 
Another obvious precursor inside game theory is the
celebrated transformation analysis of equivalent games 
with imperfect information by Thompson \cite{thompson1997equivalence} (refined by
Elmes and Reny in \cite{elmes1994strategic}), which is close to power equivalence. 
But game theory also has comparative discussions of the 
information available in extensive forms and in strategic 
normal forms \cite{mail1994}, 
a style of analysis that remains to be connected to our 
representation theorems and logics for different levels 
of describing games. 

Finally, it should be said that further
intuitions of game equivalence emerge once we consider
players' preferences, so that game equivalence can also
refer to correlations between available equilibria. This further level is beyond the scope
of this paper, but it would be a natural next step to take.

\section{Powers }

\subsection{Powers and power equivalence}
We begin  by reviewing a standard logical notion of game equivalence in terms of \emph{powers} of the players. For a standard overview of the basic concepts of game theory, see \cite{osborne1994course}. For more on equivalence of extensive games, see \cite{bonanno1992set,elmes1994strategic,thompson1997equivalence}.

\begin{definition}
A \emph{tree} $\mathcal{T}$ is a prefix closed subset of $\mathbb{N}^*$, subject to the condition that if $w \cdot j \in \mathcal{T}$ and $i < j$ then $w \cdot i \in \mathcal{T}$ as well. The  empty word $\varepsilon$ is the \emph{root} of the tree.
\end{definition}
\begin{definition}
An \emph{extensive game} $\game$ for a finite set of players $A$ with outcomes in the set $O$ is a tuple $(\mathcal{T},t,o,\Pi)$ where $\mathcal{T}$ is a finite tree, $t$ a map from $\mathcal{T}$ to $A$, $o$ a map from branches of $\mathcal{T}$ to $O$, and $\Pi$ a partition of $\mathcal{T}$ subject to the following condition: for any pair $w,v$ within the same partition cell of $\Pi$, $w$ and $v$ have the same number of children in $\mathcal{T}$, and furthermore $t(n) = t(n')$. If all partition cells of $\Pi$ are singletons we call $\game$ a game of perfect information, and we omit $\Pi$. 

Maximal branches of $\mathcal{T}$ will also be called  \emph{full matches}, and prefixes of maximal branches are called \emph{partial matches}.

 A \emph{strategy} for player $a \in A$ is a map $\sigma : t^{-1}[a] \to \mathbb{N}$ where $w \cdot \sigma(w) $ is a child of $w$ for each $w$ with $t(w) = a$, and $\sigma(w) = \sigma(w')$ whenever $w,w'$ are in the same partition cell in $\Pi$. A \emph{strategy profile} is a tuple $(\sigma_a)_{a \in A}$ of one strategy for each player in $A$. A strategy profile $p$ completely determines a full match, and hence a leaf of the game tree, so we can speak of the outcome of $p$, denoting it by $o(p)$.  Generally, we say that a full match $m$ of $\game$  is \emph{guided} by the strategy $\sigma$ for $a$ if for every prefix $w$ of $m$ such that $t(w) = a$, $\sigma(w)$ is also a prefix of $m$. $\mathsf{Match}(\sigma)$ is the set of $\sigma$-guided matches. 
\end{definition}
We denote the set of games for players $A$ with outcomes  in $O$ as $\mathbb{G}(A,O)$. For two-player games we call the players by $\plrone$ (Alice) and $\plrtwo$ (Bob). We set $\overline{\eloi} = \abel$ and $\overline{\abel} = \eloi$. 

Note that we have not attributed payoffs to matches in a game or preferences over the outcomes, but rather (and more generally) simply outcomes from some fixed chosen set. In this sense we are dealing with \emph{game forms} rather than proper games. We return to the issue of preferences in Section 7.1.

\begin{definition}
Let $\game = (\mathcal{T},t,o,\Pi)$ be a game with outcomes in $O$.
A set $P \subseteq O$ is a \emph{power} of player $a \in A$ in the game $\game$ if there is a strategy $\sigma$ for $a$ in $G$ such that $o(m) \in P$ for every $\sigma$-guided match $m$. Given a player $a \in A$ we let $P_a(\game)$ denote the set of powers of $a $ in $\game$.

Two games $\game_1,\game_2 \in \mathbb{G}(A,O)$ are \emph{power equivalent} if for all $a \in A$: $P_a(\game_1) = P_a(\game_2)$.  We denote this  by $\game_1 \sim \game_2$. If $P_a(\game_1) = P_a(\game_2)$ for some specific $a \in A$, we write $\game_1 \sim_a \game_2$.
%
\end{definition} 

Every game $\game$ in $\mathbb{G}(A,O)$  gives rise to a tuple $(P_a(\game))_{a \in A}$ of subsets of $O$, which represents a crucial aspect of social scenarios: the abilities of participants to force outcomes. 

Powers in two-player games, our focus in what follows, are characterized by three formal properties, for a set of outcomes $O$, and a pair $F_\eloi,F_\abel$ of families of subsets of $O$:
\smallskip
\begin{description}
\item[Non-emptiness] For all $u \in W$ there are $Z,Z' \subseteq W$ such that $(u,Z) \in F_\eloi$ and $(u,Z')\in F_\abel$.
\item[Monotonicity] If $P \in F_\eloi$ ($P \in F_\abel$) and $P \subseteq Q \subseteq O$, then $Q \in F_\eloi$ ($Q \in F_\abel$).
 \item[Consistency] If $P \in F_\eloi$ and $Q \in F_\abel$, then $P \cap Q \neq \emptyset$.
\end{description}

\begin{theorem}
\label{knownrepresentation}
The families $F_\eloi,F_\abel \subseteq \psf (O)$ satisfy the Non-emptiness, Monotonicity and Consistency properties if, and only if, there exists a game $\game \in \mathbb{G}(\{\eloi,\abel\},O)$ such that $F_\eloi = P_\eloi(\game)$ and $F_\abel = P_\abel(\game)$. 
\end{theorem}

A representation theorem also holds for perfect information games, with this additional property:
\smallskip
\begin{description}
\item[Determinacy] For all sets $P \subseteq O$, either $P \in F_\eloi$ or \\
$O\setminus P \in F_\abel$. 
\end{description}

\begin{theorem}
The families $F_\eloi,F_\abel \subseteq \psf (O)$ satisfy the Non-emptiness, Monotonicity, Consistency and Determinacy properties if, and only if, there exists a perfect information game $\game \in \mathbb{G}(\{\eloi,\abel\},O)$ such that $F_\eloi = P_\eloi(\game)$ and $F_\abel = P_\abel(\game)$. 
\end{theorem}
For proofs of these results we refer to \cite{vB14}.

\subsection{Neighborhood logic and bisimilarity }

The use of neighborhood semantics to interpret a propositional dynamic logic of powers in determined games dates back to \cite{parikh85}. Here, we review a modal logic GL for powers in two-player games from \cite{vB14}, which drops the game constructions (for these, see Section 6) as well as determinacy, referring explicitly to powers of separate players in the syntax.

Given a set of propositional variables $\mathsf{Prop}$, the syntax of GL is given by the following grammar:
  
$$\varphi := p  \in \mathsf{Prop} \mid   \varphi \wedge \varphi \mid \neg \varphi \mid \gal \varphi \mid \gbo \varphi $$

\medskip

The semantics for this logic uses neighborhood models that assign each player a neighborhood relation representing the powers of that player relative to each world. Of course, we must impose suitable constraints to ensure that these actually behave as powers of players in some game. The representation result Theorem \ref{knownrepresentation} tells us what these  should be:

\begin{definition}
 A  \emph{game frame} is a triple $(W,R_\eloi,R_\abel)$ such that $W$ is a set and $R_\player \subseteq W \times \psf W$ for each player $\player \in \{\eloi,\abel\}$, and such that for all  $u \in W$ the pair $R_\eloi[u],R_\abel[u]$ satisfies the Non-emptiness, Monotonicity and Consistency conditions (with $W$ viewed as the set of outcomes $O$, so that $R_\eloi[u],R_\abel[u]$ make up two families of sets of outcomes). 
A \emph{game model} is a game frame together with a valuation $V : \mathsf{Prop} \to \psf W$.
\end{definition}
We define the interpretations of all formulas in a game model $\mdl = (W,R_\eloi,R_\abel,V)$ as follows:
\\\\
\; (a) $\tset{p} = V(p)$, \; (b) $\tset{\varphi \wedge \psi} = \tset{\varphi} \cap \tset{\psi}$, \; (c) $\tset{\neg \varphi} = W \setminus \tset{\varphi}$,
\\\\
\; and crucially, for the modality: (d) $\tset{\Gpl \varphi} = R^{-1}_\player[\tset{\varphi}]$. Note that the Monotonicity condition makes this equivalent to: $u \in \tset{\Gpl \varphi} $ iff there exists $Z \subseteq \tset{\varphi}$ with $u R_\player Z$.

\smallskip
We write $\model,v \Vdash \varphi$ for $v \in \tset{\varphi}$, and $\vDash \varphi$ (`$\varphi$ is \emph{valid}') if, for every game model $\model$ and $v \in W$, we have $\model,v \Vdash \varphi$. 

Game models come with a natural notion of bisimulation:
\begin{definition}
Let $\mdl = (W,R,V)$, $\mdl' = (W',R',V')$ be game models. 
The relation $B \subseteq W \times W'$ is said to be a \emph{power bisimulation} if, for all $uBu'$  and each $\player \in \{\eloi,\abel\}$, we have:
\begin{description}
\item[Forth] For all $Z$ such that $u R_\player Z$, there exists a $Z'$ such that $u' (R_\player') Z'$ and the following condition holds:
\item[Forth-Back] For all $v' \in Z'$ there is some $v \in Z$ such that $v B v'$.
\item[Back] For all $Z'$ such that $u' R_\player Z'$ there is some $Z$ such that $u R_\player Z$ and the following condition holds:
\item[Back-Forth] For all $v \in Z$ there is some $v' \in Z' $ such that $v B v'$.
\end{description}
We say that pointed game models $\mdl,w$ and $\mathfrak{N},v$ are \emph{power bisimilar}, written  $\mdl,w \bisimilar \mathfrak{N},v$, if there is a  power bisimulation $B$ between $\mdl$ and $\mathfrak{N}$ such that $wBv$.
\end{definition}
All formulas of GL are invariant for power bisimilarity:
\begin{proposition}
If $\mdl,w \bisimilar \mathfrak{N},v$ then $\mdl,w \Vdash \varphi$ iff $\mathfrak{N},v \Vdash \varphi$, for each formula $\varphi$ of GL.
\end{proposition}
The logic GL can be axiomatized by a simple extension of monotone (multi-)modal logic. Here is a version using axiom schemata and a rule of replacement of equivalents:

\subsubsection*{Axioms for GL}
\begin{description}
\item[Non-Em] $\Gpl\top$
\item[Mon] $\gpl \varphi \rightarrow \gpl  (\varphi \vee \psi)$
\item[Cons] $\Gpl \varphi \rightarrow \neg [\adversary] \neg \varphi$
\end{description}
\subsubsection*{Proof rules}
\begin{description}
\item[MP] $$\frac{\varphi \rightarrow \psi \quad \quad \varphi}{\psi}$$
\item[RE] $$\frac{\varphi \leftrightarrow \psi \quad\quad \theta}{\theta[\varphi/\psi]}$$
where $\theta[\varphi/\psi]$ is the result of substituting some \\ occurrences of the formula $\psi$ by $\varphi$ in $\theta$.
\end{description}
\smallskip
We denote this system of axioms by $\mathbf{GL}$ and write $\mathbf{GL} \vdash \varphi$ to say that the formula $\varphi$ is provable in this axiom system.
\begin{theorem}
\label{completenessbasic}
The logic $\mathbf{GL}$ is sound and complete for validity on game frames.
\end{theorem}

The completeness proof is an exercise involving a straightforward canonical model construction, which we omit. Furthermore, $\mathbf{GL}$ is decidable and has the finite model property.

\section{Rethinking powers and  game equivalence}

\subsection{From powers to strategic equivalence}
\label{s:strat-equivalence}

Power equivalence, while a natural and simple notion of game equivalence, is relatively coarse. In particular, it misses much of the \emph{interactive} nature of games. To illustrate what we mean by this, here is an example from \cite{vB14}.

\smallskip

Consider the two games depicted in Figure \ref{twogames}. In both games, each player can perform two actions ``left'' and ``right'', and there are three possible outcomes $1,2,3$. If Alice moves left, then the outcome is $1$ regardless of the action chosen by Bob, but if Alice moves right, the outcome depends on the actions of Bob: if Bob moves left, the outcome is $2$, otherwise $3$. The difference lies in which player moves first. In the figure, the game where Alice chooses first is depicted to the left, and the game where Bob chooses first to the right.  

\vspace{-0.2cm}

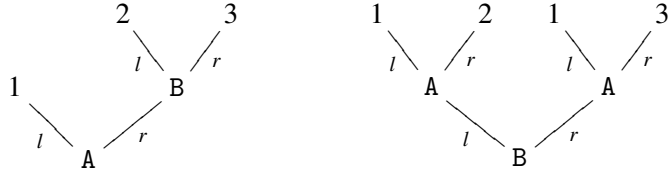
\begin{figure}[H]
\[
\xymatrixrowsep{0.5cm}
\xymatrixcolsep{0.05cm}
 \xymatrix{& & & &  & & 2 & & & & 3    \\
 & &   1    & & & & & &    \abel \ar@{-}^l[ull]\ar@{-}_r[urr] & &    \\
 & & & & & \eloi \ar@{-}^l[ulll] \ar@{-}_r[urrr] & &  & &  &  }
\quad\quad\quad \quad
\xymatrixrowsep{0.5cm}
\xymatrixcolsep{0.05cm}
 \xymatrix{1 & & & & 2 & & 1 & & & & 3    \\
 & &   \eloi \ar@{-}^l[ull] \ar@{-}_r[urr]    & & & & & &    \eloi \ar@{-}^l[ull]\ar@{-}_r[urr] & &    \\
 & & & & & \abel \ar@{-}^l[ulll] \ar@{-}_r[urrr] & &  & &  &  }
\]
\caption{
\label{twogames}
Two power equivalent games.}
\end{figure}

\vspace{-0.2cm}

It is easy to see that each player has the same powers in both games, and this is the basis for standard game logics (the games represent two sides of a standard propositional distribution law). But the interaction of the players looks different: in the right game, $\eloi$ has an obvious strategy for which the possible outcomes are precisely $1$ and $2$. But in the left game, the only way that $\eloi$ can exclude the outcome $3$ is to go left at the start of the game, making $1$ the only possible outcome. Thus, it is doubtful if one should see these games as equivalent. 

Another way of phrasing the difference is this. The two games differ if we think of powers more `socially' as what a player is going to force while at the same time recording which choices are left intentionally to \emph{the other player}. That is, both players have a say, and the notion of power becomes oriented toward both players, more in the spirit of game-theoretic equilibrium. This intuition can be made a bit more precise if we bring in a standard game-theoretic device. Let us display the  strategic forms of the two games, with rows corresponding to strategies for $\eloi$ and columns  strategies of $\abel$:

\smallskip

\[
\begin{array}{ | c | c | }
  \hline 1 & 1 \\ \hline
 2 & 3 \\ \hline
\end{array}
\quad\quad
\begin{array}{ | c | c | }
  \hline 1 & 1 \\ \hline
 2 & 1 \\ \hline
1 & 3 \\ \hline
2 & 3 \\ \hline
\end{array}
\]

\medskip

Looking at the yields of columns and rows, the above difference is clear. Here is a finer notion of game equivalence, inspired by the `matrix logic' of \cite{van2011toward}:
 
\begin{definition}
Let $\game_1$ and $\game_2$ be two-player games over the set of outcomes $O$. A \emph{strategy profile bisimulation} between these two games is a relation $R \subseteq P_1 \times P_2$, where $P_1$ is the set of strategy profiles of $\game_1$ and $P_2$ is the set of strategy profiles of $\game_2$, such that if $(\sigma_1,\tau_1) R (\sigma_2,\tau_2)$, then:
\smallskip
\begin{description}
\item[Atomic] $o_1(\sigma_1,\tau_1) = o_2(\sigma_2,\tau_2)$,
\item[Forth$(\eloi)$] For all strategies $\sigma_1'$ for $\eloi$ in $\game_1$ there is some strategy $\sigma_2'$ in $\game_2$ with $(\sigma_1',\tau_1) R (\sigma_2',\tau_2)$,
\item[Back$(\eloi)$]  For all strategies $\sigma_2'$ for $\eloi$ in $\game_2$ there is some strategy $\sigma_1'$ in $\game_1$ with $(\sigma_1',\tau_1) R (\sigma_2',\tau_2)$ 
\item[Forth$(\abel)$] For all strategies $\tau_1'$ for $\abel$ in $\game_1$ there is some strategy $\tau_2'$ in $\game_2$ with $(\sigma_1,\tau_1') R (\sigma_2,\tau_2')$,
\item[Back$(\abel)$]  For all strategies $\tau_2'$ for $\abel$ in $\game_2$ there is some strategy $\tau_1'$ in $\game_1$ with $(\sigma_1,\tau_1') R (\sigma_2,\tau_2')$
\end{description}
\smallskip
We call $\game_1$ and $\game_2$ are \emph{strategic form equivalent} if there is a strategy profile bisimulation $R$ between them, relating every profile in $P_1$ to some profile in $P_2$, and vice versa.
\end{definition}  
This equivalence concept is more fine-grained than power equivalence. In particular, the power equivalent games displayed in Figure \ref{twogames} are not strategic form equivalent, as can be seen by inspecting their matrix forms. However, this approach sacrifices much of the logical simplicity of power equivalence. We therefore proceed to modify the notion of power itself in line with the above strategic form perspective.

\subsection{Basic powers and strong equivalence}

Our proposed new game equivalence works as follows.

\begin{definition}
Let $\game$ be any game in $\mathbb{G}(A,O)$, let $a \in A$. A power $P \subseteq O$ is said to be a  \emph{basic power} for $a$ in $\game$ if there is a strategy $\sigma$ for $a$ in $\game$ such that $P = \{ o(m) \mid m \in \mathsf{Match}(\sigma) \}$.
The set of all basic powers of $a$ in $\game$ is denoted by $B_a(\game)$. 
\end{definition}
\begin{definition}
Two games $\game_1$ and $\game_2$ are \emph{strongly power equivalent}, written $\game_1 \simeq \game_2$, iff $B_a(\game_1) = B_a(\game_2)$ for all $a \in A$. We write $ \game_1 \simeq_a \game_2$ to say that $B_a(\game_1) = B_a(\game_2)$. 
\end{definition} 
Strong power equivalence is more fine-grained than power equivalence: the two games in Figure \ref{twogames} are not strongly power equivalent. It also retains a connection to strategic forms.

\begin{proposition}
Any two strategic form equivalent games are strongly power equivalent, and any two strongly power equivalent games are power equivalent.
\end{proposition}

All inclusions are strict here. Here are two games that are strongly power equivalent but not strategic form equivalent -- displayed in strategic form with outcome set $\{0,1\}$:
\smallskip
\[
\begin{array}{ | c | c | c |}
  \hline 0 & 1  & 0\\ \hline
 1 & 0 & 0  \\ \hline
 0 & 0 & 0  \\ \hline
\end{array}
\quad\quad
\begin{array}{ | c | c | c | }
  \hline 1 & 1 & 0\\ \hline
 0 & 0 & 0  \\ \hline
\end{array}
\]

\medskip

In the matrix on the right, the profile in the middle upper square is not bisimilar with any  profile on the left.

\medskip

As a prelude to our later logical analysis, we generalize strong power equivalence to games with different outcomes. Given $\game_1 \in \mathbb{G}(A,O_1)$ and $\game_2 \in \mathbb{G}(A,O_2)$, $R  \subseteq O_1 \times O_2$ is a \emph{strategy bisimulation} between $\game_1$ and $\game_2$ if, for all $a \in A$:

\smallskip

\begin{description}
\item[Forth] For all $Z_1 \in B_a(\game_1) $, there exists $Z_2 \in B_a(\game_2)$ such that $Z_1 \tilde{R} Z_2$,
\item[Back] For all $Z_2 \in B_a(\game_2) $, there exists $Z_1 \in B_a(\game_1)$ such that $Z_1 \tilde{R} Z_2$,
\end{description}

\smallskip

\noindent where $\tilde{R}$ is the \emph{Egli-Milner lifting} of $R$. I.e., $Z \tilde{R} Z' $ if, for all $x \in Z$, there is $x' \in Z'$ with $x R x'$, and vice versa. It is clear that strong power equivalence is a special case of this.

\medskip

Having proposed our new notion of game equivalence, we now determine its basic properties. This is the content of the following representation theorem for basic powers. Obviously, the earlier monotonicity condition has to be dropped, since it typically fails on our new reading of powers as also offering choices to the other player. On the other hand, this role for the other player also validates a new condition that did not hold before. Consider any pair $F_\eloi,F_\abel \subseteq \psf(O)$:
\smallskip
\begin{description}
\item[Instantiatedness] Given $P \in F_\eloi$ ($P \in F_\abel$): for any $x \in P$, there is some $P'\in F_\abel$ ($P' \in F_\eloi$) with $x \in P'$. 
\end{description}
%
%

\begin{theorem}
\label{representation}
Suppose $F_\eloi,F_\abel \subseteq \psf(O)$. Then the pair $(F_\eloi,F_\abel)$ satisfies the Non-emptiness, Instantiatedness and Consistency conditions if, and only if, there exists a game $\game$ such that $F_\eloi = B_\eloi(\game)$ and  $F_\abel = B_\abel(\game)$.
\end{theorem}

We can give a more compact statement of these conditions in terms of the above Egli-Milner lifting. Then the Instantiatedness and Consistency conditions together become: for all $P \in F_\eloi$, we have $P \widetilde{\in} F_\abel$, and for all $P \in F_\abel$, we have $P \widetilde{\in} F_\eloi$. 

In order to prove the theorem, it will be convenient to work with games in strategic form:
\begin{definition}
A \emph{strategic form two-player game} with outcomes in $O$ is a tuple $(\Sigma_\eloi,\Sigma_\abel,o)$ such that $\Sigma_\eloi$ and $\Sigma_\abel$ are non-empty sets (interpreted as \emph{strategy sets} for each player) and $o : \Sigma_\eloi \times \Sigma_\abel \to O$ is the outcome map. An element of $\Sigma_\eloi \times \Sigma_\abel$ is called a \emph{strategy profile}.
\end{definition} 
We can define powers and instantiated powers for strategic form games in the expected manner. In particular an \emph{instantiated power} for Player $\eloi$ is a subset $P$ of $O$ such that, for some strategy $\sigma \in \Sigma_\eloi$:
$$P = \{u \in O \mid o(\sigma,\sigma') = u \text{ for some } \sigma' \in \Sigma_\abel\}$$
and dually for $\abel$. Since every extensive game has a strategic normal form, and conversely every strategic form game is the strategic normal form of some extensive game of imperfect information, we can work with strategic and extensive games interchangeably (see for example \cite{osborne1994course}). 

It is straightforward to check that the conditions Non-emptiness, Consistency and Instantiatedness hold for the instantiated powers of any game. For the converse, let $F_\eloi,F_\abel \subseteq \psf(O)$ be given, and suppose all three conditions hold for the pair $(F_\eloi,F_\abel)$. We shall construct a game $\game$ such that the instantiated powers of each player $\player$ in $\game$ coincides with the set $F_\player$. We construct $\game$ as a strategic form game, as follows:
\begin{itemize}
\item The set $\Sigma_\abel$ of strategies for $\abel$ is just the set $F_\abel \times O \times \{0,1\}$.
\item The set $\Sigma_\eloi$ of strategies for $\eloi$ is defined as the collection of all maps $c : \Sigma_\abel \to O$ such that:
\begin{itemize}
\item $c(Z,u,j) \in Z$ for all $Z \in F_\abel$, $u \in O$ and $j \in \{0,1\}$, and
\item the image $c[\Sigma_\abel] $ of the set $ \Sigma_\abel$ under the map $c$ is a member of $F_\eloi$.
\end{itemize}
\item The outcome map $o$ sends a strategy profile $(c,(Z,u,j)) \in \Sigma_\eloi \times \Sigma_\abel$ to $c(Z,u,j)$.
\end{itemize}
The set $\Sigma_\abel$ is non-empty by the Non-emptiness condition, and it will follow from Claim \ref{fonesubset} below that $\Sigma_\eloi$ is also non-empty.
The appearance of the set $O \times \{0,1\}$ in this construction is merely a way to create ``enough copies'' of each set in $F_\abel$ to make sure that certain suitable strategies for $\eloi$ can be defined. In particular, it allows us to establish the following claim:
\begin{claim}
\label{fonesubset}
For every set $Z \in F_\eloi$, there exists a strategy $c$ for  $\eloi$ in $\game$ such that $c[F_\abel \times O \times \{0,1\}] = Z$.
\end{claim}
\begin{proof}[Proof of Claim \ref{fonesubset}]
Suppose $Z \in F_\eloi$. For every $u \in Z$, there exists some $Z' \in F_\abel$ with $u \in Z'$, by the Instantiatedness property. So we can define a choice function $g : Z \to F_\abel$ such that for each $u \in Z$ we have $u \in g(u)$. We can modify this $g$ to obtain a map $$g' : Z \to F_\abel \times O \times \{0,1\}$$ by mapping $u \in Z$ to the triple $(g(u),u,0)$. We can now define the strategy $c$ as follows: given a triple $(g(u),u,0)$ in the image of $Z$ under the map $g'$, we set  $c(g(u),u,0) = u$. For every triple $(Z',u',k)$ not in the image of $g'$, we set $c(Z',u',k)$ to be some arbitrary element of $Z \cap Z'$, which exists by the Consistency condition. Clearly we get that the image of the set $F_\abel \times O \times \{0,1\}$ under the map $c$ is equal to $Z$. Furthermore, since $c(Z',u',k) \in Z'$ for each triple $(Z',u',k)$, we get that $c$ is a legitimate strategy for $\eloi$ in $\game$.
\end{proof}
A second claim that we will need is the following:
\begin{claim}
\label{oinftwo}
Let $Z \in F_\abel$,  $u \in Z$, $u' \in O$ and let $j \in \{0,1\}$. Then there exists a strategy $c$ for  $\eloi$ in $\game$ such that $c(Z,u',j) = u$.
\end{claim}
\begin{proof}[Proof of Claim \ref{oinftwo}]
By the Instantiatedness property there exists some $Z'\in F_\eloi$ such that $u \in Z'$. But by Claim \ref{fonesubset} there exists some $c$ which is a legitimate move for $\eloi$ in $\game$, such that $c[F_\abel \times O \times \{0,1\}] = Z'$. In fact, by inspection of the proof of Claim \ref{fonesubset}, we see that we may pick $c$ so that $c[F_\abel \times O \times \{0\}] = Z'$ as well. Now define the map $c'$ as follows: if $j = 0$ then define $c'$ to be like $c$ except that $c'(Z,u',0) = u$ and $c'(Z,u',1) = c(Z,u',0)$. If $j = 1$, then define $c'$ to be like $c$ except that $c'(Z,u',1) = u$. In either case, we still have $c'[F_\abel \times O \times \{0,1\}] = Z'$, and so we see that $c'$ is a legitimate strategy for $\eloi$. Since $c' (Z,u',j) = u$, we are done. 
\end{proof}
We now show that every instantiated power for either player $\player$ is an element of $F_\player$, and vice versa. We have four different inclusions to prove, two for each player. 

Suppose that $Z \in F_\eloi$. By Claim \ref{fonesubset} there exists a strategy $c$ for  $\eloi$ such that $c[F_\abel \times O \times \{0,1\}] = Z$. But since the possible strategies for $\abel$ are exactly the members of the set $F_\abel \times O \times \{0,1\}$, it follows that $Z$ is an instantiated power for $\eloi$.

Conversely, suppose that $Z$ is an instantiated power for $\eloi$ in $\game$. Then there exists a strategy $c$ for $\eloi$ such that the possible outcomes consistent with $c$ are precisely the members of $Z$. It follows that $c[F_2 \times O \times \{0,1\}] = Z$. But since the strategies for $\eloi$ are subject to the constraint that $c[F_2 \times O \times \{0,1\}] \in F_\eloi$, it follows that we must have $Z \in F_\eloi$.

To prove the inclusions for $\abel$ we need the following claim: 
\begin{claim}
\label{samez}
Let $(Z,u,k) \in F_\abel \times O \times \{0,1\}$, considered as a strategy for $\abel$. Then the possible outcomes consistent with this strategy are precisely the members of $Z$, i.e. $Z = \{c(Z,u,k) \mid c \in \Sigma_\eloi\}$.
\end{claim}
\begin{proof}[Proof of Claim \ref{samez}]
Let $Z'$ denote the set of all possible outcomes consistent with the strategy $(Z,u,k)$. It is clear that $Z' \subseteq Z$, since every possible outcome $u'$ consistent with the strategy  $(Z,u,k)$ is of the form $c(Z,u,k)$ for some strategy $c \in F_\eloi$, which means that $u'= c(Z,u,k) \in Z$. Conversely, if $u' \in Z$ then by Claim \ref{oinftwo} there exists a strategy $c$ for $\eloi$ such that $c(Z,u,k) = u'$. So $u'$ is an outcome that is consistent with the strategy $(Z,u,k)$, hence $u'\in Z'$. So we get $Z = Z'$, and the proof is finished.
\end{proof}
We can now easily prove both of the inclusions for $\abel$: if $Z \in F_\abel$, then $(Z,u,0)$ is a legitimate strategy for  $\abel$ for any arbitrarily chosen $u \in O$, and it now follows directly from Claim \ref{samez} that $Z$ is an instantiated power for $\abel$. Conversely, if $Z$ is an instantiated power for  $\abel$ in $\game$ then there is some strategy $(Z',u,k)$ for $\abel$ witnessing this. By Claim \ref{samez} we get $Z = Z'$, and since   $(Z',u,k)$ is a strategy for $\abel$ we have $Z' \in F_\abel$. So $Z \in F_\abel$, and we are done.

From this proof, we can also read off the following result:

\begin{theorem}
Our properties of basic powers also capture the powers computed from rows and columns of  matrix games.
\end{theorem}

\smallskip

At present we do not have a representation theorem for basic powers in the special case of perfect information games. It is easy to find additional conditions on powers that hold in this setting, but we  have not yet found a complete set.

\section{The logic of basic powers: instantial game logic}
\label{inlsection}

What is the game logic that goes with strong power equivalence? Our earlier notion of strategy bisimulation points the way. It resembles the \emph{instantial neighborhood bisimulations}  introduced in \cite{vb16} as the invariance underlying \emph{instantial neighborhood logic}.
Accordingly, we now introduce a logic for games at this level of structure. The syntax  of \emph{instantial game logic} IGL is given by the following grammar:   
\smallskip
$$\varphi := p  \in \mathsf{Prop} \mid   \varphi \wedge \varphi \mid \neg \varphi \mid \gal(\Psi ;\varphi) \mid \gbo(\Psi;\varphi)$$
\smallskip
where $\Psi$ ranges over finite sets of formulas of IGL. We sometimes write $\gpl(\psi_1,...,\psi_n;\varphi)$ rather than $\gpl(\{\psi_1,...,\psi_n\};\varphi)$, $\gpl(\psi;\varphi)$ for $\gpl(\{\psi\};\varphi)$, and $\gpl \varphi$ for $\gpl(\emptyset;\varphi)$.  

\medskip

The semantics of IGL, as for GL, uses neighborhood models. However, the constraints are different, since we are now dealing with \emph{basic} powers. These constraints come from our representation result for basic powers Theorem \ref{representation}.

\begin{definition}
 An  \emph{instantial game frame} $(W,R_\eloi,R_\abel)$ is a triple with $W$ a set, and $R_\player \subseteq W \times \psf W$ for each player $\player \in \{\eloi,\abel\}$, where for all  $u \in W$, the pair $(R_\eloi[u],R_\abel[u])$ satisfies Non-emptiness, Instantiatedness and Consistency. \emph{Instantial game models} then add a valuation for propositional variables.
\end{definition}
The key clause in the truth definition in instantiated game models $\mdl = (W,R_\eloi,R_\abel,V)$ runs as follows:

\medskip
 $u \in \tset{\Gpl (\psi_1,...,\psi_k;\varphi)}$ iff there is some $Z \subseteq W$ such that
\smallskip

$(u,Z) \in R_\player$ and $Z\subseteq \tset{\varphi}$, $Z \cap \tset{\psi_i} \neq \emptyset$ for $i \in \{1,...,k\}$.
\medskip

If we interpret formulas as `outcomes', then we see why IGL is a natural language for basic powers: the formula $\gpl(\Psi;\bigvee \Psi)$ says that $\Psi$ is a basic power for the player $\player$, while the weaker formula $\gpl \bigvee \Psi$ says that $\Psi$ is simply a power. 

\medskip

Instantial models come with a notion of bisimulation which stands to strong power equivalence as standard neighborhood bisimulations stands to power equivalence:
\begin{definition}
Let $\mdl = (W,R,V)$ and $\mdl' = (W',R',V')$ be any neighborhood models. 
The relation $B \subseteq W \times W'$ is said to be an \emph{instantial neighborhood bisimulation} if, for all $uBu'$  and $\player \in \{\eloi,\abel\}$, we have:
\begin{description}
\item[Forth] For all $Z$ such that $u R_\player Z$, there is some $Z'$ such that $u' (R_\player') Z'$ and the following conditions hold:
\item[Forth-Back] For all $v' \in Z'$ there is some $v \in Z$ such that $v B v'$.
\item[Forth-Forth] For all $v \in Z$ there is some $v' \in Z' $ such that $v B v'$.
\item[Back] For all $Z'$ such that $u' R_\player Z'$ there is some $Z$ such that $u R_\player Z$ and the following condition holds:
\item[Back-Forth] For all $v \in Z$ there is some $v' \in Z' $ such that $v B v'$.
\item[Back-Back] For all $v' \in Z'$ there is some $v \in Z$ such that $v B v'$.
\end{description}
Pointed instantial game models $\mdl,w$ and $\mathfrak{N},v$ are \emph{instantial neighborhood bisimilar},  $\mdl,w \bisimilar \mathfrak{N},v$, if some instantial neighborhood bisimulation $B$ between $\mdl$ and $\mathfrak{N}$ has $wBv$.
\end{definition}
Formulas of IGL are invariant for instantial bisimilarity:
\begin{proposition}
If $\mdl,w \bisimilar \mathfrak{N},v$ then $\mdl,w \Vdash \varphi$ iff $\mathfrak{N},v \Vdash \varphi$, for each formula $\varphi$ of IGL.
\end{proposition}

More can be said about the model theory of instantial neighborhood simulation, but the present facts suffice here.
 
\section{Axiomatizing IGL}

In this section we axiomatize the valid formulas of IGL, thus pinning down the modal logic of basic powers. Our system is a gentle modification of instantial neighborhood logic.  
\subsubsection*{IGL axioms.}
\begin{description}
\item[Mon] $\gpl(\psi_1,...,\psi_n;\varphi) \rightarrow \gpl(\psi_1 \vee \alpha_1,...,\psi_n\vee \alpha_n; \varphi \vee \beta)$
\item[Weak] $\gpl(\Psi;\varphi) \rightarrow \gpl(\Psi';\varphi)$ for $\Psi' \subseteq \Psi$
\item[Un] $\gpl(\psi_1,...,\psi_n; \varphi) \rightarrow \gpl(\psi_1 \wedge \varphi,...,\psi_n\wedge \varphi; \varphi)$
\item[Lem] $\gpl(\Psi;\varphi) \rightarrow \gpl(\Psi \cup \{\gamma\};\varphi) \vee \gpl(\Psi;\varphi \wedge \neg \gamma)$
\item[Bot] $\neg \gpl(\bot;\varphi)$
\end{description}
\subsubsection*{Axioms for frame constraints}
\begin{description}
\item[Non-Em] $\Gpl\top$
\item[Inst] $\Gpl(\psi;\top) \leftrightarrow [\adversary](\psi;\top)$
\item[Cons] $\Gpl \varphi \rightarrow \neg [\adversary] \neg \varphi$
\end{description}
\subsubsection*{Proof rules.}
\begin{description}
\item[MP] $$\frac{\varphi \rightarrow \psi \quad \quad \varphi}{\psi}$$
\item[RE] $$\frac{\varphi \leftrightarrow \psi \quad\quad \theta}{\theta[\varphi/\psi]}$$
\end{description}
\medskip
We denote this system of axioms by $\mathbf{IGL}$ and write $\mathbf{IGL} \vdash \varphi$ to say that the formula $\varphi$ is provable in this axiom system. 
\begin{theorem}
\label{completeness}
The system $\mathbf{IGL} $ is sound and complete for validity over instantial game models.
\end{theorem} 

The soundness part of this result is checked case by case, and the easy argument is omitted. 
The completeness proof proceeds via a normal form argument, following an idea in \cite{vB16}. The adaptation to the present setting is not trivial however, since we have to deal with the new frame constraints of Non-emptines, Consistency and Instantiatedness. The main contribution here is thus to prove that the model construction satisfies these constraints. We outline the key parts of the proof below.
\begin{definition}
The \emph{modal depth} of a formula is defined inductively by:
\\\\
- $d(p) = 0$ 
\\\\
- $d(\neg \varphi) = d(\varphi)$ 
\\\\
- $d(\varphi \wedge \psi) = \mathtt{max}(d(\varphi),d(\psi))$ 
\\\\
- $d(\gpl (\Gamma;\varphi)) = \mathtt{max}(d[\Gamma \cup \{\varphi\}]) + 1$
\end{definition}
\begin{definition}
Given a finite set of propositional variables $Q$, a formula $\varphi$ is said to be a \emph{$Q$-formula} if all propositional variables appearing in $\varphi$ belong to $Q$. 

Given $k \in \omega$ and a finite set $Q$ of propositional variables, a \emph{$(Q,k)$-description} is a consistent $Q$-formula $\varphi$ of modal depth $\leq k$, such that for any $Q$-formula $\theta$ of depth $\leq k$, we have $\varphi \vdash \theta$ or $\varphi \vdash \neg \theta$.
\end{definition}
Note that there are at most finitely many $Q$-formulas of depth $\leq k$, given that $Q$ is finite. We omit the (standard) argument for this.

The key lemma for the completeness proof is the following:
\begin{lemma}
\label{normalformlemma}
Let $\gpl(\Gamma;\varphi)$ be a formula such that $\mathtt{\max}(d[\Gamma \cup \{\varphi\}]) \leq k$ and let $Q$ be a finite set of propositional variables containing all variables appearing in this formula. Then $\gpl(\Gamma;\varphi)$ is provably equivalent to some disjunction of the form:
$$\bigvee_{i \in I} \gpl(\Theta_i;\bigvee \Theta_i)$$
where $I$ is a finite set and for each $i \in I$, $\Theta_i$ is a finite set of $(Q,k)$-descriptions, such that:
\\

- every member of $\Theta_i$ provably entails $\varphi$, and

- every member of $\Gamma$ is provably entailed by some member of $\Theta_i$.
\end{lemma}
For a proof of this lemma, see \cite{vB16}.

Now fix  a finite set of propositional variables $Q$. Given a $Q$-formula $\varphi$ let $\widehat{\varphi}$ denote the equivalence class of the formula under provable equivalence. For a finite set of formulas $\Gamma$ set 
$$\widehat{\Gamma} = \{\widehat{\varphi} \mid \varphi \in \Gamma\}$$
We construct a neighborhood model $\model = (W,R,V)$ as follows:
\begin{itemize}
\item $W = \{(\widehat{\varphi},k) \mid \varphi \text{ is a $(Q,k)$-description and $k < \omega$}\}$
\item For a player $\player$ let $R_\player$ be the union of the sets $$\{((\widehat{\varphi},k + 1),\widehat{\Gamma} \times \{k\}) \in W \times \psf(W) \mid \varphi \vdash \gpl (\widehat{\Gamma};\bigvee \widehat{\Gamma}) \} $$
and
$$\{((\widehat{\varphi},0),W) \mid (\widehat{\varphi},0) \in W\} $$
\item Finally, for any propositional variable $p$, set $V(p) = \{\widehat{\varphi} \mid \varphi \vdash p\}$ if $p \in Q$, $V(p) = \emptyset $ otherwise.
\end{itemize}
Note that this is well defined, i.e. whether $(\widehat{\varphi},\widehat{\Gamma}) \in R_\player$ is independent of the choice of witnesses $\varphi,\Gamma$ of the equivalence classes. The following lemma can be proved exactly as in \cite{vb16}, and we refer to that paper for the details:
\begin{lemma}[Truth lemma]
\label{truthlemmasf}
Let $\model$ be constructed as above, and let $\psi$ be any basic formula of modal depth $\leq k$ whose propositional variables all belong to $Q$ and such that all game terms appearing in $\psi$ belong to $\tau$. Then for every $(Q,\tau,k)$-description $\varphi$, we have:
$$\model,(\widehat{\varphi},k)\Vdash \psi \text{ iff } \varphi \vdash \psi$$
\end{lemma}

The addition we need to make here is the following lemma:
\begin{lemma}
\label{isamodel}
The structure $\model$ construced above is a game model, i.e. it satisfies the Non-emptiness, Consistency and Instantiatedness constraints.
\end{lemma}
\begin{proof}
First, note that all the conditions hold for the image of each relation on an element of $W$ of the form $(\widehat{\varphi},0)$. So we can focus on the images of relations of the form $R_\player$ on states of the form $(\widehat{\varphi},k+1)$ for some $k$.

The Non-emptiness condition is proved straightforwardly using the axiom  (Non-Em), we leave this to the reader. For Instantiatedness, suppose that 
$$ (\statekp,\widehat{\Theta}\times \{k\}) \in R_\eloi$$
By definition, we get $\varphi \vdash \Gal(\Theta;\bigvee \Theta)$. Pick an element $(\widehat{\theta},k) \in \Theta \times \{k\}$. By (Weak) and (Mon) we get $ \Gal(\Theta;\bigvee \Theta) \vdash \Gal(\theta;\top)$, so $\varphi \vdash \Gal(\theta;\top)$. By the axiom (Inst) we get $\varphi \vdash \Gbo(\theta;\top)$ as well. Since $\varphi$ is a $(Q,k + 1)$-description, we can derive from Lemma \ref{normalformlemma} that there is some set $\Psi$ of $(Q,k)$-descriptions such that $\varphi \vdash \Gbo(\Psi;\bigvee \Psi)$, and such that there exists some $\psi \in \Psi$ with $\psi\vdash \theta$. But since $\psi,\theta$ are both $(Q,k)$-descriptions, clearly this means that $\widehat{\theta} = \widehat{\psi}$, so $(\widehat{\theta},k) = (\widehat{\psi},k)$. But this means that we get $(\statekp,\widehat{\Psi}\times \{k\}) \in R_\abel$ and $(\widehat{\theta},k) \in \widehat{\Psi}\times\{k\}$ as required. The converse direction is proved in the same manner.

For the Consistency condition, suppose that $(\statekp,\widehat{\Theta}\times\{k\}) \in R_\eloi$ and  $(\statekp,\widehat{\Theta'}\times\{k\}) \in R_\abel$. It is straightforward to prove, using that $\Theta$ and $\Theta'$ are both sets of $(Q,k)$-descriptions, that if $\widehat{\Theta}\times\{k\}$ does not intersect  $\widehat{\Theta'}\times\{k\}$ then in fact $\bigvee \Theta' \rightarrow \neg \bigvee \Theta$.  But we have $\varphi \vdash \gal(\Theta;\bigvee \Theta)$, hence $\varphi \vdash \gal \bigvee \Theta$ by the axiom schema (Weak). Furthermore we have:
$$\varphi \vdash \Gbo(\Theta';\bigvee\Theta') \vdash \Gbo\bigvee\Theta' \vdash \Gbo \neg \bigvee\Theta $$
But then $\varphi \vdash  \Gal \bigvee \Theta \wedge\Gbo\neg \bigvee\Theta $, and it follows from the axiom schema (Cons) that $\varphi$ cannot be consistent, which contradicts our assumption that $\varphi$ was a $(Q,k+1)$-description.
\end{proof}

Combining Lemmas \ref{isamodel} and \ref{truthlemmasf} with the easy observation that any consistent basic formula of depth $\leq k$, variables in $Q$ and atomic games among $\tau$ is provably entailed by some $(Q,k)$-description\footnote{This follows from Lindenbaum's lemma together with the observation that there are at most finitely many formulas of depth $\leq k$, variables in $Q$ and game terms among $\tau$ up to provable equivalence, so we can take a conjunction of all representatives of each equivalence class of such formulas belonging to a given maximal consistent set.}, we obtain Theorem \ref{completeness}.

As a corollary to this proof, we get:
\begin{theorem}
The logic IGL is decidable and has the effective finite model property.
\end{theorem}

IGL is a high-level logic of basic powers in social interaction. The reader may find it of interest to see what the above axioms say when read as statements about games.

\section{Adding game operations}
Our third contribution in this paper concerns the addition of structure to games, in the form of natural game operations. 

\subsection{Game algebra of strong powers}
In this section we use some basic concepts of universal algebra, see for example \cite{BS81}. For simplicity, we restrict attention to finite games, so that $\mathbb{G}(\{\eloi,\abel\},O)$ is now the set of finite games with outcomes in $O$. Thus the outcome map of a game $\game$ can be viewed a map $o$ from the leaves in $\game$ into $O$.

Consider a set of games on a fixed set of outcomes $O$. We define operations in a standard manner, with binary $+,\times$ corresponding to choice for $\eloi,\abel$ respectively, and a unary operation $-$ for game dual (`role switch'). The game $\game_1 + \game_2$ ($\game_1 \times \game_2$) is defined as follows: let $\game_1 = (\mathcal{T}_1,t_1,o_1,\Pi_1)$ and let $\game_2 = (\mathcal{T}_2,t_2,o_2,\Pi_2)$.  We first construct the tree $\mathcal{T}'$ by adding a new root $r$ with two successors, and the left successor is the root of a subtree isomorphic with $\mathcal{T}_1$ via a fixed isomorphism $i_1$, the right successor is the root of a subtree isomorphic with $\mathcal{T}_2$ via a fixed isomorphism $i_2$. The turn function $t'$ is defined by setting $t'(r) = \eloi$  ($t'(r) = \abel$). For a node $u$ in the subtree corresponding to the left successor of the root $r$ we set $t'(u) = t_1(i_1(u))$ and similarly for a node $u$ in the subtree corresponding to the right successor of the root $r$ we set $t'(u) = t_2(i_2(u))$. The outcome map $o'$ is defined by setting $o'(l) = o_1(i_1(l))$ for a leaf in the subtree corresponding to the left successor of $r$, and $o'(l) = o_2(i_2(l))$ for a leaf in the subtree corresponding to the right successor of $r$. We define a partition $\Pi'$ by setting $$\Pi' = \{\{r\}\} \cup \{i_1^{-1}[Z] \mid Z \in \Pi_1\} \cup \{i_2^{-1}[Z] \mid Z \in \Pi_2\}.$$
The game $\game_1 + \game_2$ ($\game_1 \times \game_2$) is then defined as $(\mathcal{T}',t',o',\Pi')$.

 The construction of $-\game_i$ is much simpler, it merely changes the turn assignment by switching players at each position, otherwise keeping everything the same. 
\begin{proposition}
Strong power equivalence is a congruence on the algebra $ \langle \mathbb{G}(\{\eloi,\abel\},O), +,\times,-\rangle$.
\end{proposition}
This motivates the following definition:
\begin{definition}
The \emph{strong algebra of games} $\mathfrak{G}$ (with outcomes $O$) is the quotient:
$$ \langle \mathbb{G}(\{\eloi,\abel\},O), +,\times,-\rangle/\simeq$$
\end{definition}

The equational theory of the algebra $\mathfrak{G}$ is of special interest here, as it can be viewed as a new weaker propositional logic, where distributivity fails, witness our example in Figure \ref{twogames}. By contrast, for standard power equivalence, this algebra is known to be a distributive \emph{de Morgan algebra}.

But with strong power equivalence, much more basic principles than distributivity fail. For instance, the operations $\times$ and and $+$ are not idempotent, the following equations fail:
 $$ x \times x = x \quad \quad \quad x + x = x$$
For the first failure, take a two-player game $\game$ in which $\eloi$ has the first move, and simply chooses between two outcomes $0,1$. Then $\{0,1\}$ is a basic power of $\eloi$ in $\game \times \game$, but not in $\game$. 

Still, many laws known from game algebra do go through. 
\begin{proposition}
The following equations hold in $\mathfrak{G}$:
\begin{description}
\item[Associativity] $x + (y + z) = (x + y) + z$, \\ $x \times (y \times z) = (x \times y) \times z$
\item[Commutativity] $x + y  = y + x$, \;  $x \times y = y \times x$
\item[Double Negation] $-- x = x$
\item[De Morgan] $-(x + y) = -x \times -y$, \; $-(x \times y) = -x + -y$
\end{description}
\end{proposition}

\subsection{Dynamic game logic for basic powers}
Our game algebra still misses one important operation, namely, \emph{sequential composition}. For this operation to make sense, we need to take a dynamic view of games as state-transforming processes, in the style of dynamic game logic (cf. \cite{parikh85}, \cite{van200720},\cite{vB14}). 
\begin{definition}
A \emph{dynamic two-player game} over a set $X$ (of ``states'') is a map $g : X \to \mathbb{G}(\{\eloi,\abel\},X)$, assigning a game with outcome set $X$ to each state in $X$. We denote the set of dynamic two-player games over $X$ by $ \mathbb{D}(\{\eloi,\abel\},X)$.
\end{definition}
The operations $+,\times$ and $-$ are naturally lifted to dynamic games by defining them component-wise in an obvious manner, and so are the relations $\sim,\simeq$ of power equivalence and strong power equivalence. Now, given dynamic games $g_1,g_2 : X \to \mathbb{G}(\{\eloi,\abel\},X)$, we can define the \emph{sequential composition} $g_1 \circ g_2$ by letting $g_1 \circ g_2(u)$ be constructed by replacing each leaf $l$ in $g_1(u)$ by a copy of the game tree $g_2(o_1(l))$, where $o_1$ is the outcome map associated with $g_1(u)$.

In this way, we get an extended game algebra. For power equivalence, the complete algebra of the propositional operations plus sequential composition has been axiomatized in \cite{gorankogamealgebra,venemagamealgebra}. For basic powers and strong equivalence in our new sense, however, this is an open problem. 
\smallskip

However, we now run into some unexpected trouble:
\begin{proposition}
Let $O = \{x,y\}$. Then the relation of strong power equivalence over $\mathbb{D}(\{\eloi,\abel\},O)$ is not a congruence with respect to sequential game composition. 
\end{proposition}

To see why this is so, consider the two perfect information games displayed in Figure 2, which have an obvious instantiantial neighborhood bisimulation between them.  The games are not strongly power equivalent, since player $\abel$ has a basic power $\{x, y\}$ in the game to the right, but not to the left. But both games can clearly be obtained as the sequential composition of strongly power equivalent games:

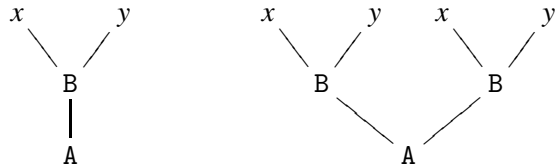
\begin{figure}[H]
\[
\xymatrixrowsep{0.5cm}
\xymatrixcolsep{0.05cm}
 \xymatrix{x & & & & y    \\
 & &   \abel    \ar@{-}[ull]\ar@{-}[urr] & &    \\
 & & \eloi \ar@{-}[u]  & &   }
\quad\quad\quad \quad
\xymatrixrowsep{0.5cm}
\xymatrixcolsep{0.05cm}
 \xymatrix{x & & & & y & & x & & & & y    \\
 & &  \abel \ar@{-}[ull] \ar@{-}[urr]    & & & & & &    \abel\ar@{-}[ull]\ar@{-}[urr] & &    \\
 & & & & & \eloi \ar@{-}[ulll] \ar@{-}[urrr] & &  & &  &  }
\]
\caption{A threat to compositionality: instantial bisimulation does not preserve basic powers.}
\end{figure}

This failure might seem a serious challenge to the compositional methodology of dynamic game logic. But when we analyze what goes wrong in the example, the reason is the \emph{functional} character of strategies. They force a unique choice at each turn, making $\abel$ too specific in the game on the left.

\medskip 
To remedy this situation, we suggest to widen the notion of a strategy to allow \emph{non-determinism}, so that strategies may constrain the moves of a player, but not determine them uniquely. This is not altogether foreign in game theory: in fact, mixed strategies can be interpreted in a similar way, and the same move has also been defended for the broader notion of a `plan' in \cite{vB14}.
\begin{definition}
A \emph{relational strategy} for player $\player$ in a game $\game = (\mathcal{T},t,o,\Pi)$ is a binary relation $\sigma$ over $\mathcal{T}$ such that:
\begin{itemize}
\item $\sigma[u] \neq \emptyset$ whenever $u \in t^{-1}[\player]$, and
\item $\sigma[u] = \sigma[v]$ if $u,v$ are in the same partition cell in $\Pi$.
\end{itemize}
\smallskip
The set $\mathsf{Match}(\sigma)$ of non-deterministic matches guided by a strategy $\sigma$ is defined in the obvious way. We say that $P \subseteq O$ is a \emph{basic relational power} of $\player$ if there is a relational strategy $\sigma$ for $\player$ in $\game$ such that $P = \{o(m) \mid m \in \mathsf{Match}(\sigma)\}$. The set of basic relational powers of $\player$ is denoted by $R_\player(\game)$. We say that $\game_1,\game_2$ are \emph{semi-strongly power equivalent} if $R_\player(\game_1) = R_\player(\game_2)$, for each $\player \in \{\eloi,\abel\}$. Finally, we write $\game_1 \equiv \game_2$ when $\game_1,\game_2$ are semi-strongly power equivalent.
\end{definition} 
The relation of semi-strong power equivalence can be lifted to an equivalence relation between dynamic games in the same component-wise manner as before. 

\medskip

The move to relational strategies does not trivialize the new equivalence concept proposed in this paper. The two games in our running example of Figure \ref{twogames} are not semi-strongly power equivalent, so this slightly coarser equivalence notion is still fine enough to make the distinction that we wanted. Furthermore, we get the result we are after:
\begin{proposition}
The relation $\equiv$ of semi-strong power equivalence over $\mathbb{D}(\{\eloi,\abel\},O)$ is a congruence with respect to the operations $+,\times,-$ plus sequential game composition $\circ$. 
\end{proposition}
\begin{definition}
The \emph{strong dynamic game algebra} is the quotient 
$ \langle \mathbb{D}(\{\eloi,\abel\},O), +,\times,-,\circ\rangle/\equiv$.
\end{definition}

In this game algebra, we get a number of interesting valid equations known from the game algebra of power equivalence:
\begin{proposition}
\label{p:equations-comp}
The following equations hold in $\mathfrak{D}$:
\begin{description}
\item[Associativity] $x \circ (y \circ z) = (x \circ y) \circ z$
\item[Dualization] $-(x \circ y) = (-x)\circ (-y)$
\item[Left Distribution] $(x + y)\circ z = (x \circ z) + (y \circ z) $
\end{description}
\end{proposition}
Note also that some equations that were not valid in the functional setting now become valid, such as idempotence: $x + x = x$. In this way the algebra is closer to the algebra of games under power equivalence, but does not collapse to it: the distribution law $ x \times (y + z) = (x \times y) + (x \times z)$ still fails, for example.

Relational strategies improve our game algebra. At the same time, our earlier results go through with suitable modifications. In particular, our representation theorem for basic powers is easily amended to capture relational basic powers. Consider the following condition on pairs $F_\eloi,F_\abel \subseteq \psf(O)$:
\smallskip
\begin{description}
\item[Union Closure] For any non-empty family of relational basic powers $M \subseteq F_\eloi$ ($M \subseteq F_\abel$), we have that $\bigcup M  \in F_\eloi$ ($\bigcup M \in F_\abel$).
\end{description}
\begin{theorem}
Let $F_\eloi,F_\abel$ be two families of subsets of $O$. Then $F_\eloi,F_\abel$ satisfy Non-emptiness, Consistency, Instantiatedness and Union Closure if, and only if there is a game $\game$ such that $F_\eloi = R_\eloi(\game)$ and $F_\abel = R_\abel(\game)$.
\end{theorem}
 
Finally, our basic game logic can also be extended with game terms to capture this algebraic reasoning, in the style of dynamic game logic, \cite{parikh85}, \cite{vB14}. We will then have instantial modalities describing basic powers of player $i$ in game $G$: 

 \medskip
$[G, i](\Psi;\varphi)$
\medskip

\noindent This formalism can be interpreted on our instantial game models, when we provide these with world-dependent basic power relations for all players. The crucial point here is that, with the earlier obstacle to compositionally overcome, we can define the power relation for a product game $\game_1 \circ \game_2$ in the following inductive manner.

\medskip
$(u,Z) \in R_{\game_1 \circ \game_2}^\player$ iff $Z = \bigcup F$, for some family $F \subseteq \psf W$  and some $Y \subseteq W$ with $(u,Y) \in R_{\game_1}^\player$ and $(Y,F) \in \widetilde{R}_{\game_2}^{\;\player}$. 
\medskip

This is the basis for obvious recursion axioms for the game operations that lead to the following result.

\begin{theorem}
The dynamic game logic of relational basic powers is completely axiomatizable.
\end{theorem}

This logic has interesting further properties that deviate from known systems, especially in its axiomatization of game iteration, which we will treat in a separate publication.

\section{Further directions}

In this final section we briefly consider some directions for future reseach, in particular:
\begin{itemize}
\item Enriching the logic with epistemic modalities to reason about imperfect information.
\item Incorporating preference into the framework of basic powers and instantial game logic.
\item Placing instantial game logic and strong power equivalence in a wider view perspective on the many possible game logics and notions of game equivalence. 
\end{itemize} 

\subsection{Imperfect information}
In this paper, we have worked with imperfect information games from the start. This raises some issues of intuitive interpretation. Imperfect information in games can arise for quite different reasons: players'  limited powers of observation of moves, but also players' uncertainty about the strategies of other players. All this brings in what players \emph{know} about the game, and a richer game logic reflecting this would have to incorporate epistemic modalities. Moreover, imperfect information sits somewhat uneasily with our game algebra, since information sets can cross between subgames, disrupting the obvious compositional structure. We have side-stepped this issue in the definition for our game operations, but at the price of dealing with only a special class of imperfect information games. Clearly, a lot remains to be clarified.

\subsection{Preference}

This paper has studied `game forms' with abstract outcomes without any specified preference ordering. A natural next step is to consider proper games in which the players have preference orders over the set of outcomes. This is necessary to connect the game equivalences we have considered here with standard game theoretic concepts such as Nash equilibrium or solution methods like the elimination of dominated strategies. 



Adding modalities for preference is a well-known device in game logics, so we could also do this. But preference does raise questions for our perspective. For instance, the games in Figure 1 have different Backward Induction  solutions with preference $2 < 1 < 3$ for $\abel$, $1 < 3 < 2$ for $\eloi$. We may have to redefine our basic powers in the presence of preference, considering only preference-optimal sets for a player. But solution methods like Backward Induction also incorporate one particular view of rationality: they make an assumption about agents, rather than being part of the neutral mathematics of the game. Perhaps we need to study game equivalences parametrized to particular types of player.

\subsection{A multitude of game logics}

This paper does not claim that there is one best level for viewing games. Extensive form, standard powers, or strategic form all have their virtues, and we have merely claimed that there is room for one more natural new option. All these levels come with their own logical languages matching the invariance relation, \cite{vB14}: relational modal logic for extensive games, modal neighborhood logic for standard powers, instantial neighborhood logic for our basic powers, and multi-modal logics accessing the different dimensions of matrix games. 

This raises a systematic question. How are all these different logics related, given natural transformations from one level of game structure to another? For instance, on finite games, the modal logic of forcing powers can be translated into a $\mu$-calculus on the underlying extensive games, and the same is true of our instantial modalities for basic powers, using the observations in Section 6. But modalities for strategic form games are not easily compared with our forcing modalities: moving across rows or columns means considering alternative strategies for players, something that would require a serious extension of our forcing language. In addition, matrix logics have surprising features that have no counterpart in our forcing logics, such as the undecidability of the full system for three players, which reflect the undecidability of the product logic $S5 \times S5 \times S5$ \cite{Maddux80}.

We believe that systematizing the total picture of game logics and their interrelations holds great interest, but we must leave this for further investigation.


{
\bibliographystyle{eptcs}
\bibliography{NBD}

\begin{thebibliography}{10}
\providecommand{\bibitemdeclare}[2]{}
\providecommand{\surnamestart}{}
\providecommand{\surnameend}{}
\providecommand{\urlprefix}{Available at }
\providecommand{\url}[1]{\texttt{#1}}
\providecommand{\href}[2]{\texttt{#2}}
\providecommand{\urlalt}[2]{\href{#1}{#2}}
\providecommand{\doi}[1]{doi:\urlalt{http://dx.doi.org/#1}{#1}}
\providecommand{\bibinfo}[2]{#2}

\bibitemdeclare{article}{van2011toward}
\bibitem{van2011toward}
\bibinfo{author}{E.~Pacuit \surnamestart J.~van Benthem\surnameend} \&
  \bibinfo{author}{O.~\surnamestart Roy\surnameend} (\bibinfo{year}{2011}):
  \emph{\bibinfo{title}{Toward a theory of play: A logical perspective on games
  and interaction}}.
\newblock {\sl \bibinfo{journal}{Games}}
  \bibinfo{volume}{2}(\bibinfo{number}{1}), pp. \bibinfo{pages}{52--86},
  \doi{10.3390/g2010052}.

\bibitemdeclare{book}{vB14}
\bibitem{vB14}
\bibinfo{author}{J.~van \surnamestart Benthem\surnameend}
  (\bibinfo{year}{2014}): \emph{\bibinfo{title}{Logic in games}}.
\newblock \bibinfo{publisher}{MIT Press, Cambridge, MA}.

\bibitemdeclare{article}{vB16}
\bibitem{vB16}
\bibinfo{author}{J.van \surnamestart Benthem\surnameend},
  \bibinfo{author}{N.~\surnamestart Bezhanishvili\surnameend},
  \bibinfo{author}{S.~\surnamestart Enqvist\surnameend} \&
  \bibinfo{author}{J.~\surnamestart Yu\surnameend} (\bibinfo{year}{2017}):
  \emph{\bibinfo{title}{Instantial Neighborhood Logic}}.
\newblock {\sl \bibinfo{journal}{The Review of Symbolic Logic}}
  \bibinfo{volume}{10}(\bibinfo{number}{1}), pp. \bibinfo{pages}{116--144},
  \doi{10.1017/S1755020316000447}.

\bibitemdeclare{article}{bonanno1992set}
\bibitem{bonanno1992set}
\bibinfo{author}{G.~\surnamestart Bonanno\surnameend} (\bibinfo{year}{1992}):
  \emph{\bibinfo{title}{Set-theoretic equivalence of extensive-form games}}.
\newblock {\sl \bibinfo{journal}{International Journal of Game Theory}}
  \bibinfo{volume}{20}(\bibinfo{number}{4}), pp. \bibinfo{pages}{429--447},
  \doi{10.1007/BF01271135}.

\bibitemdeclare{book}{BS81}
\bibitem{BS81}
\bibinfo{author}{R.~\surnamestart Burris\surnameend} \&
  \bibinfo{author}{H.~\surnamestart Sankappanavar\surnameend}
  (\bibinfo{year}{1981}): \emph{\bibinfo{title}{A Course in Universal
  Algebra}}.
\newblock \bibinfo{publisher}{Springer}, \doi{10.1007/978-1-4613-8130-3}.

\bibitemdeclare{article}{elmes1994strategic}
\bibitem{elmes1994strategic}
\bibinfo{author}{S.~\surnamestart Elmes\surnameend} \& \bibinfo{author}{P.~J.
  \surnamestart Reny\surnameend} (\bibinfo{year}{1994}):
  \emph{\bibinfo{title}{On the strategic equivalence of extensive form games}}.
\newblock {\sl \bibinfo{journal}{Journal of Economic Theory}}
  \bibinfo{volume}{62}(\bibinfo{number}{1}), pp. \bibinfo{pages}{1--23},
  \doi{10.1006/jeth.1994.1001}.

\bibitemdeclare{article}{mail1994}
\bibitem{mail1994}
\bibinfo{author}{L.~Samuelson \surnamestart G.~J.~Mailath\surnameend} \&
  \bibinfo{author}{J.~M. \surnamestart Swinkels\surnameend}
  (\bibinfo{year}{1994}): \emph{\bibinfo{title}{Normal form structures in
  extensive form games}}.
\newblock {\sl \bibinfo{journal}{Journal of Economic Theory}}
  \bibinfo{volume}{64}(\bibinfo{number}{2}), pp. \bibinfo{pages}{325--371},
  \doi{10.1006/jeth.1994.1072}.

\bibitemdeclare{article}{gorankogamealgebra}
\bibitem{gorankogamealgebra}
\bibinfo{author}{V.~\surnamestart Goranko\surnameend} (\bibinfo{year}{2003}):
  \emph{\bibinfo{title}{The basic algebra of game equivalences}}.
\newblock {\sl \bibinfo{journal}{Studia Logica}}
  \bibinfo{volume}{75}(\bibinfo{number}{2}), pp. \bibinfo{pages}{221--238},
  \doi{10.1023/A:1027311011342}.

\bibitemdeclare{incollection}{van200720}
\bibitem{van200720}
\bibinfo{author}{W.~\surnamestart van~der Hoek\surnameend} \&
  \bibinfo{author}{M.~\surnamestart Pauly\surnameend} (\bibinfo{year}{2007}):
  \emph{\bibinfo{title}{Modal logic for games and information}}.
\newblock In \bibinfo{editor}{van Benthem~J. \surnamestart
  Blackburn\surnameend, P.} \& \bibinfo{editor}{F.~\surnamestart
  Wolter\surnameend}, editors: {\sl \bibinfo{booktitle}{Handbook of modal
  logic}}, {\sl \bibinfo{series}{Studies in Logic and Practical
  Reasoning}}~\bibinfo{volume}{3}, \bibinfo{publisher}{Elsevier}, pp.
  \bibinfo{pages}{1077--1148}, \doi{10.1016/S1570-2464(07)80023-1}.

\bibitemdeclare{book}{bergstra2001handbook}
\bibitem{bergstra2001handbook}
\bibinfo{author}{A.~Ponse \surnamestart J.~A.~Bergstra\surnameend} \&
  \bibinfo{author}{S.~A. \surnamestart Smolka\surnameend}
  (\bibinfo{year}{2001}): \emph{\bibinfo{title}{Handbook of process algebra}}.
\newblock \bibinfo{publisher}{Elsevier}.

\bibitemdeclare{article}{Maddux80}
\bibitem{Maddux80}
\bibinfo{author}{R.~\surnamestart Maddux\surnameend} (\bibinfo{year}{1980}):
  \emph{\bibinfo{title}{The Equational Theory of CA\({}_{\mbox{3}}\) is
  Undecidable}}.
\newblock {\sl \bibinfo{journal}{J. Symb. Log.}}
  \bibinfo{volume}{45}(\bibinfo{number}{2}), pp. \bibinfo{pages}{311--316},
  \doi{10.2307/2273191}.

\bibitemdeclare{article}{osborne1994course}
\bibitem{osborne1994course}
\bibinfo{author}{M.~\surnamestart Osborne\surnameend} \&
  \bibinfo{author}{A.~\surnamestart Rubinstein\surnameend}
  (\bibinfo{year}{1994}): \emph{\bibinfo{title}{A Course in Game Theory}}.
\newblock {\sl \bibinfo{journal}{MIT Press Books}} \bibinfo{volume}{1}.
  \doi{10.2307/2554642}

\bibitemdeclare{incollection}{parikh85}
\bibitem{parikh85}
\bibinfo{author}{R.~\surnamestart Parikh\surnameend} (\bibinfo{year}{1985}):
  \emph{\bibinfo{title}{The logic of games and its applications}}.
\newblock In: {\sl \bibinfo{booktitle}{Topics in the theory of computation
  ({B}orgholm, 1983)}}, {\sl \bibinfo{series}{North-Holland Math. Stud.}}
  \bibinfo{volume}{102}, \bibinfo{publisher}{North-Holland, Amsterdam}, pp.
  \bibinfo{pages}{111--139}, \doi{10.1016/S0304-0208(08)73078-0}.

\bibitemdeclare{article}{thompson1997equivalence}
\bibitem{thompson1997equivalence}
\bibinfo{author}{F.~\surnamestart Thompson\surnameend} (\bibinfo{year}{1997}):
  \emph{\bibinfo{title}{Equivalence of games in extensive form}}.
\newblock {\sl \bibinfo{journal}{Classics in game theory}}
  \bibinfo{volume}{36}, \doi{10.1006/jeth.1994.1001}.

\bibitemdeclare{article}{venemagamealgebra}
\bibitem{venemagamealgebra}
\bibinfo{author}{Y.~\surnamestart Venema\surnameend} (\bibinfo{year}{2003}):
  \emph{\bibinfo{title}{Representation of game algebras}}.
\newblock {\sl \bibinfo{journal}{Studia Logica}}
  \bibinfo{volume}{75}(\bibinfo{number}{2}), pp. \bibinfo{pages}{239--256},
  \doi{10.1023/A:1027363028181}.

\end{thebibliography}
}

\end{document}